\documentclass[twocolumn,aps,pra,superscriptaddress]{revtex4-1}

\usepackage{amsmath}
\usepackage{graphicx}
\usepackage[utf8]{inputenc}
\usepackage[normalem]{ulem}
\usepackage{dcolumn}
\usepackage{epsfig}
\usepackage{bm}
\usepackage{array}
\usepackage[english]{babel}
\usepackage{hyperref}
\usepackage{algorithm,algorithmic}
\usepackage{tikz}
\usetikzlibrary{calc,arrows,chains,matrix,positioning,scopes}

\hypersetup{
colorlinks=true,
citecolor=blue,
linkcolor=red,
urlcolor=blue,
 pdfmenubar=true
}

\usepackage{hyperref}
\usepackage{graphicx}
\usepackage{amsmath}
\usepackage{amsfonts}
\usepackage{amssymb}
\usepackage{xcolor}
\definecolor{teal}{rgb}{0.,0.5,.5}
\usepackage{bbm}
\usepackage{calrsfs}
\usepackage{dutchcal}
\usepackage{amsthm}
\usepackage{lipsum}

\AtBeginDocument{%
    \newwrite\bibnotes
    \def\bibnotesext{Notes.bib}
    \immediate\openout\bibnotes=\jobname\bibnotesext
    \immediate\write\bibnotes{@CONTROL{REVTEX41Control}}
    \immediate\write\bibnotes{@CONTROL{%
    apsrev41Control,author="08",editor="1",pages="1",title="0",year="1"}}
     \if@filesw
     \immediate\write\@auxout{\string\citation{apsrev41Control}}%
    \fi
}%

\newtheorem{theorem}{Theorem}

\newcommand{\be}{\begin{equation}}
\newcommand{\ee}{\end{equation}}

\newcommand{\beq}{\begin{eqnarray}}
\newcommand{\eeq}{\end{eqnarray}}

\begin{document}

\title{Genuine Multipartite Nonlocality with Causal-Diagram Postselection}

\author{Valentin~Gebhart}\affiliation{QSTAR, INO-CNR and LENS, Largo Enrico Fermi 2, 50125 Firenze, Italy}
\affiliation{Universit\`a degli Studi di Napoli ”Federico II”, Via Cinthia 21, 80126 Napoli, Italy}

\author{Luca~Pezz\`e}%
\affiliation{QSTAR, INO-CNR and LENS, Largo Enrico Fermi 2, 50125 Firenze, Italy}

\author{Augusto~Smerzi}%
\affiliation{QSTAR, INO-CNR and LENS, Largo Enrico Fermi 2, 50125 Firenze, Italy}

\begin{abstract}
The generation and verification of genuine multipartite nonlocality (GMN) is of central interest for both fundamental research and quantum technological applications, such as quantum privacy. To demonstrate GMN in measurement data, the statistics are commonly postselected by neglecting undesired data. Until now, valid postselection strategies have been restricted to local postselection. A general postselection that is decided after communication between parties can mimic nonlocality, even though the complete data are local. Here, we establish conditions under which GMN is demonstrable even if observations are postselected collectively. Intriguingly, certain postselection strategies that require communication among several parties still offer a demonstration of GMN shared between all parties. The results are derived using the causal structure of the experiment and the no-signalling condition imposed by relativity. Finally, we apply our results to show that genuine three-partite nonlocality can be created with independent particle sources. 
\end{abstract}

\maketitle

{\it Introduction.---}
Bell nonlocality \cite{bell1964,bell1976} is one of the most intriguing discoveries in modern physics. 
Besides its heavily discussed fundamental significance and implications \cite{brunner2014,wiseman2014,scarani2019book}, several technological applications have been developed in fields such as communication \cite{brukner2004,buhrman2010}, quantum cryptography \cite{ekert1991,barrett2005,acin2007,masanes2011,ekert2014}, certified random number generation \cite{pironio2010,colbeck2011}, and quantum computation \cite{raussendorf2001,raussendorf2003}. 
Growing interest is experienced by the field of multipartite nonlocality \cite{hillery1999,brunner2014,svetlichny1987,bancal2009,bancal2013,epping2017,pivoluska2018,ribeiro2018}. 
Here, genuine multipartite nonlocality (GMN), a subclass of multipartite nonlocality, plays a central role.
Genuinely multipartite nonlocal correlations cannot be described by nonlocal correlations confined to different groups of subsystems but require collective nonlocal correlations between all subsystems \cite{svetlichny1987,bancal2009,bancal2013}. 
This stronger form of nonlocality is the key ingredient for many future quantum technologies such as the quantum internet \cite{kimble2008,wehner2018,murta2020} and device-independent multipartite quantum key distributions \cite{hillery1999,epping2017,pivoluska2018,ribeiro2018}, and serves as a detection of genuine multipartite entanglement \cite{werner1989,carvalho2004,horodecki2009,pezze2009,hyllus2012,toth2012,pezze2018,szalay2019,ren2021}.

Imagine a group of $n$ experimental parties that have performed an experiment together. Now they want to examine if their observed results demonstrate the presence of GMN by the violation of a Bell inequality \cite{bell1964,clauser1969}, cf. Fig . \ref{fig:multipartite}.
A first test of the Bell inequalities using the complete measurement statistics does not yield any violation. 
It is known that a common postselection strategy that can be decided locally, i.e., each party knows whether to keep or neglect its measurement result without knowledge of other parties, can be used to verify Bell nonlocality \cite{sciarrino2011}. 
Say that even a local postselection of results does not suffice for a violation of the Bell inequalities. 
Generally, the more data is ignored, the more Bell inequalities can be violated, cf. Figs. \ref{fig:multipartite}(b) and \ref{fig:multipartite}(c). 
However, the correlations could be created by the postselection bias \cite{pearl2016}:
a postselection that is decided collectively by all experimental parties can potentially mimic nonlocal behaviour even if the underlying statistics can be described by local hidden-variable models \cite{aerts1999,cabello2009,decaro1994,lima2010}. 
Can the parties employ strategies beyond local postselection to verify genuine multipartite Bell nonlocality of their correlations? 

\begin{figure}[b]
    \centering
    \includegraphics[width=.9\linewidth]{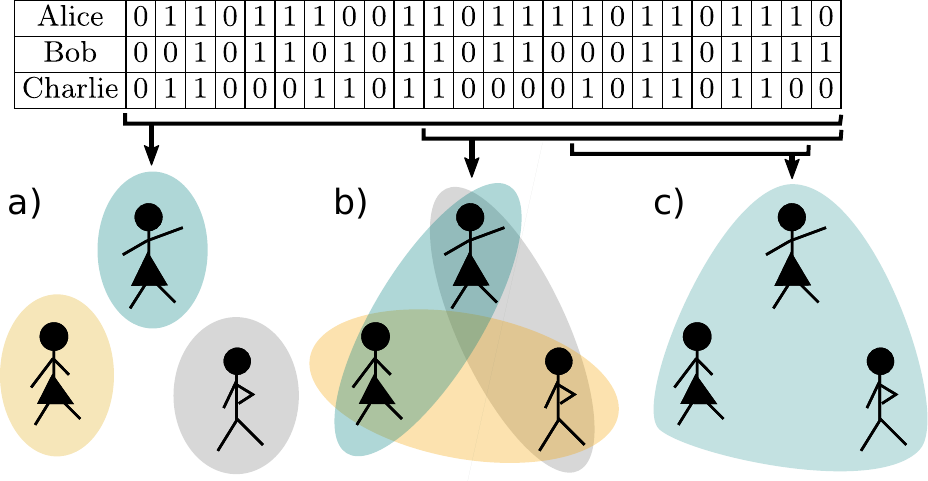}
    \caption{After postselecting their experimental observations, three parties could observe that they share (a) local correlations, (b) multipartite (but not genuine multipartite) nonlocality, or (c) genuine multipartite nonlocality (GMN). When can they be sure that their postselection did not create fake correlations?}
    \label{fig:multipartite}
\end{figure}

An instance of this problem affects a proposal by Yurke and Stoler (YS) \cite{yurke1992a} to create EPR effects \cite{einstein1935} of the GHZ type \cite{greenberger1990,mermin1990} from independent particle sources. 
Unlike common approaches to create nonlocal behaviour or entanglement by an interaction between subsystems that are send to the parties, YS make use of the (bosonic or fermionic) Hong-Ou-Mandel effect \cite{hong1987} to create nonlocality \cite{yurke1992a,yurke1992b}.
A generalization of Ref.~\cite{yurke1992a} became a standard method to create optical GHZ states \cite{greenberger1990,zukowski1993,zukowski2000,pan2012}. The EPR effects in Ref.~\cite{yurke1992a} suffice to exclude local hidden-variable models. However, the necessary postselection to demonstrate GMN cannot be decided locally. Thus, the YS scheme did not show GMN until now.

In this work, we introduce postselection strategies beyond local postselection that can be employed to demonstrate GMN. 
Our analysis of the multipartite Bell scenario provides sufficient conditions such that a collective postselection is valid.
Particularly, in the $n$-partite case, a postselection is valid if it can be equivalently decided after exclusion of any $\left \lfloor{n/2}\right \rfloor$ parties from the decision, where $\left \lfloor{.}\right \rfloor$ is the floor function. 
Thus, somewhat surprisingly, statistics that were postselected through communication of several parties can still serve as a certification of GMN among all parties.
In the three-partite case, the postselection condition simplifies to an all-but-one principle similar to Ref.~\cite{blasiak2020}, where causal diagrams are used to safely postselect statistics to verify general (non-genuine) multipartite nonlocality. 
Here, we employ causal diagrams for hybrid local-nonlocal hidden-variable models (instead of local hidden-variable models \cite{blasiak2020}) to prove valid postselections for GMN. Additionally, we explicitly use the no-signalling principle dictated by relativity. 
The analysis is performed using causal diagrams \cite{pearl2016,blasiak2020} that extend the language of statistics to include causal explanations of multivariate data.
In contrast to the common use of causal diagrams as an explanatory tool in quantum physics, the current work and Ref.~\cite{blasiak2020} show that causal diagrams can be exploited to derive new theorems. 
We finally apply our results to the YS setup \cite{yurke1992a} to show that GMN can be created from independent particle sources.

{\it Main results.---}
Bell nonlocality can be certified if the measured statistics violate a Bell inequality that was derived assuming local realism and free will. 
By the latter assumptions, the joint probability distribution of experimental results can be written as a local hidden-variable model \cite{bell1964,clauser1969}, yielding conditions on the statistics termed Bell inequalities.
To derive Bell inequalities that distinguish between partial nonlocality and GMN, the local hidden-variable model is replaced by a hybrid local-nonlocal hidden-variable model \cite{svetlichny1987,gallego2012,bancal2013,new_defs}. 
A violation of these inequalities demonstrates the presence of GMN. 

Commonly, measured statistics are postselected to obtain statistics that violate Bell inequalities. 
However, postselection potentially creates additional correlations due to the postselection bias \cite{pearle1970,pearl2016,blasiak2020} that can mimic nonlocal behaviour. 
In the following, we introduce postselection strategies that are valid to demonstrate GMN.

First, consider a three-partite Bell scenario. 
The hybrid hidden-variable model asserts that, given that the three parties Alice, Bob and Charlie measure observables $x$, $y$ and $z$, the probability for outcomes $a$, $b$ and $c$ is given by \cite{svetlichny1987,gallego2012,bancal2013}
\begin{align}
& P_{abc|xyz}=\sum_{\lambda_1\in\Lambda_1} P_{\lambda_1}P_{bc|yz\lambda_1}P_{a|x\lambda_1} \notag \\ &+ \sum_{\lambda_2\in\Lambda_2}P_{\lambda_2}P_{ac|xz\lambda_2}P_{b|y\lambda_2} + \sum_{\lambda_3\in\Lambda_3}P_{\lambda_3}P_{ab|xy\lambda_3}P_{c|z\lambda_3},\label{eq:hybrid}
\end{align}
where $\sum_{\lambda\in\Lambda}P_{\lambda}=1$, $\Lambda=\Lambda_1 \cup \Lambda_2 \cup \Lambda_3$. 
Here, we divided the hidden variables $\Lambda$ into subsets $\Lambda_i$ indicating which two parties share nonlocal correlations for a given $\lambda$. The free will assumption was used to write $P_{\lambda|xyz}=P_{\lambda}$, i.e., measurement choices are independent of the hidden variables. Furthermore, the correlations in Eq.~\eqref{eq:hybrid} must fulfill the no-signalling principle \cite{bancal2013}, e.g.,
\begin{equation}
    P_{a|xyz} = P_{a|x}.\label{eq:no-signalling}
\end{equation}
This ensures that no party can send information to others instantaneously by choice of the measurement setting. 
While valid Bell inequalities for GMN can be derived without demanding the no-signalling principle \cite{svetlichny1987}, stronger Bell inequalities can be proven including no-signalling (or one-way signalling) conditions \cite{bancal2013}. Our results only hold if the hybrid hidden-variable model fulfills the no-signalling (or one-way signalling) conditions. 
A diagram describing all possible causal relations of the different random variables is shown in Fig.~\ref{fig:diagrams}(a). Each variable is represented as a capital letter while their possible values are denoted as lowercase letters. 
Solid arrows describe possible causal influences along the arrows' directions \cite{pearl2016}. 
We emphasize that some causal influences are restricted by the no-signalling conditions and cannot be described by any classical causal model without fine-tuning \cite{wood2015,allen2017}.
For instance, while there might be causal influences from $Y$ to $B$ and from $B$ to $A$, there is no causal influence from $Y$ to $A$, cf. Eq.~\eqref{eq:no-signalling}. By conditioning on a particular $\Lambda_i$, the causal diagram can be restricted, see Fig.~\ref{fig:diagrams}(b) for $\Lambda_3$. In the following, boxes around variables indicate that the variables are conditioned on. 

\begin{figure}[htp]
\centering
\includegraphics[width=.9\linewidth]{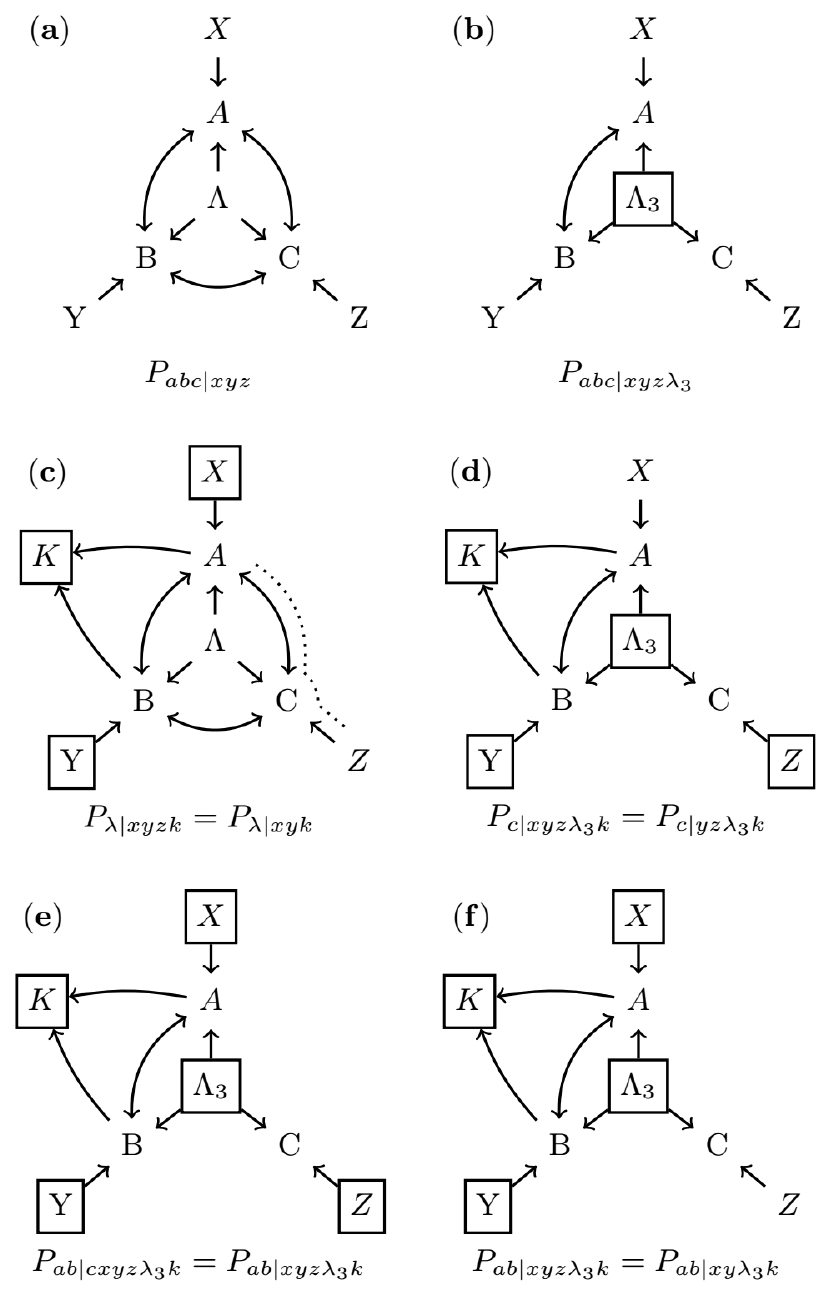}
\caption{(a) Causal diagram of the three-partite hybrid hidden-variable model \eqref{eq:hybrid}. The causal relations obey the no-signalling conditions \eqref{eq:no-signalling}. (b) Causal diagram of the subensemble $\Lambda_3$, allowing for correlations between $A$ and $B$. (c-f) Causal diagrams representing different steps of the proof of Theorem 1 where the postselection ($K$) is decided by two parties. Solid arrows represent possible causal influences between variables. Conditioned variables are marked with a box. In (c), we indicate a fine-tuning condition from the no-signalling principle as a dotted line.}
\label{fig:diagrams}
\end{figure}

Now consider a postselection (represented as a binary variable $K$) of the observed results. We expand the postselected statistics in terms of the local hidden variables, 
\begin{equation}
P_{abc|xyzk}= \sum_{\lambda} P_{\lambda|xyzk} P_{abc|xyz\lambda k}. 
\end{equation}
This exposes that a sufficient condition that $P_{abc|xyzk}$ fulfills a Bell inequality derived from Eq.~\eqref{eq:hybrid}, is that $P_{abc|xyzk}$ factorizes in a similar way: if 
\begin{align}
\mathbf{I}& \quad P_{\lambda|xyzk}=P_{\lambda|k}, \notag \\
\mathbf{IIc}& \quad P_{abc|xyz \lambda_3 k}=P_{ab|xy \lambda_3 k}P_{c|z \lambda_3 k},\notag 
\end{align}
with similar conditions $\mathbf{IIa}$ and $\mathbf{IIb}$ for $\Lambda_1$ and $\Lambda_2$, the postselected statistics are valid to demonstrate GMN. One could also require that the postselected statistics fulfill the no-signalling principle \eqref{eq:no-signalling} if the no-signalling conditions are needed to derive the Bell inequality of interest \cite{bancal2013}.

We now focus on postselection strategies that can be equivalently decided by any subset of the experimental parties of a certain minimum size. For three parties, we consider a postselection $K$ that can be equivalently decided by any two parties, implying that 
\begin{equation}
    P^{(AB)}_{abc|xyzk} = P^{(AC)}_{abc|xyzk} = P^{(BC)}_{abc|xyzk},
\end{equation}
where, e.g., $P^{(AB)}_{abc|xyzk}$ denotes the postselected conditional probability when the postselection is decided by Alice and Bob. Thus, the postselected distribution coincides for whichever two parties reconcile to decide it, and we simply write $P_{abc|xyzk}$ in the following. This decision equivalence requires global conditions on the possible experimental results such that the latter become partially redundant. An example that we will further discuss below are experiments where, for postselection, each party should find a certain number of particles and the total number of particles is conserved.

The central tool to prove our results is causal inference and the d-separation tool set \cite{pearl2016,blasiak2020}: 
given a causal diagram that connects different variables (nodes) with causal relations (arrows), only certain dependencies between variables are possible \cite{pearl2016}. In short, two variables can only be dependent if they are connected by a path. The d-separation rules dictate which paths are blocked when conditioning on other variables of the diagram and read: (1) a path is blocked if there is a collider along the path, i.e., a variable at which causal arrows collide, (2) conditioning on a non-collider along the path blocks the path, (3) conditioning on a collider (or its descendant) along the path unblocks the path. We can now prove our first main result. 
\begin{theorem}
 In the three-partite Bell scenario, a postselection that can be equivalently decided by any two (all-but-one) parties is valid for verification of genuine three-partite nonlocality. 
\end{theorem}
\begin{proof}
Assume that the postselection $K$ can be equivalently decided by any two parties. Thus, we can add to the causal diagrams of Fig.~\ref{fig:diagrams}(a,b) the variable $K$ with causal influences from any two parties, e.g., $A$ and $B$, see Fig.~\ref{fig:diagrams}(c). All of the resulting diagrams are valid and can be used in the proof. 
To prove condition $\mathbf{I}$, we first show $P_{\lambda|xyzk} = P_{\lambda|xyk}$: In Fig.~\ref{fig:diagrams}(c), we condition on $X$, $Y$, and $K$ (indicated as boxes) and examine all possible paths between $\Lambda$ and $Z$. The direct path $Z\rightarrow C \leftarrow \Lambda$ is blocked because $C$ is a collider (that is not conditioned on) along this path. Consider the path $Z\rightarrow C\rightarrow A\rightarrow K \leftarrow B \leftarrow \Lambda$. If there were general causal influences from $Z$ to $A$, this path would be open because the collider $K$ is conditioned on. However, due to the no-signalling conditions, $Z$ cannot influence $A$ and this path is blocked. This blocked path segment is indicated as a dotted line in Fig.~\ref{fig:diagrams}(c). Similarly, all other paths between $\Lambda$ and $Z$ are blocked. Note that the conditioning on $X$ and/or $Y$ can also be removed without unblocking any path. By using similar diagrams as in Fig.~\ref{fig:diagrams}(c) but with $K$ decided by $A$ and $C$ (and conditioning on $X$ and $K$), one shows that $P_{\lambda|xyk} = P_{\lambda|xk}$. Finally, a diagram with $K$ decided by $B$ and $C$ (and conditioning only on $K$) yields $P_{\lambda|xk} = P_{\lambda|k}$ and condition $\mathbf{I}$ follows. For the corresponding causal diagrams, see the Supplementary Material \cite{SI}. No-signalling conditions on the postselected statistics such as, e.g., $P_{a|xyz\lambda k} = P_{a|x\lambda k}$, can be proven in a similar fashion. 
Note that condition $\mathbf{I}$ holds true even if the no-signalling condition is replaced by weaker one-way signalling condition \cite{bancal2013}.\\
To show condition $\mathbf{IIc}$, we first use the chain rule to write 
\begin{equation}
    P_{abc|xyz \lambda_3 k} = P_{ab|cxyz \lambda_3 k}P_{c|xyz \lambda_3 k}.
\end{equation}
We have that $P_{c|xyz \lambda_3 k}=P_{c|yz \lambda_3 k}$, see Fig.~\ref{fig:diagrams}(d), and similarly $P_{c|yz \lambda_3 k}=P_{c|z \lambda_3 k}$. One also has that $P_{ab|cxyz \lambda_3 k}=P_{ab|xyz \lambda_3 k}$ (Fig.~\ref{fig:diagrams}(e)) and $P_{ab|xy \lambda_3 k}$ (Fig.~\ref{fig:diagrams}(f)). This yields condition $\mathbf{IIc}$. Conditions $\mathbf{IIa}$ and $\mathbf{IIb}$ can be shown using the relevant diagrams. 
Note that to prove conditions $\mathbf{II}$, we do not need the no-signalling conditions. 
\end{proof}

Now let us turn to conditions for valid postselections to demonstrate genuine $n$-partite nonlocality for $n>3$. First, consider the case $n=4$. The corresponding local-nonlocal hidden-variable model consists of subensembles that allow nonlocal correlations among at most three parties. Similar to above, an all-but-one postselection is valid when applied to subensembles for which three parties share nonlocality. However, for subensembles in which two pairs share bipartite nonlocality, an all-but-one postselection (i.e., a postselection decided by three parties) can create a postselection bias: the postselected distribution of these subensembles generally does not factorize into two pairs of parties. This insight is discussed in detail in \cite{SI}. Therefore, for $n=4$, we can only exclude a postselection bias if the postselection can be decided by any two parties. Generally, we can prove the following theorem. A detailed proof is given in \cite{SI}. 

\begin{theorem}
 In the $n$-partite Bell scenario, a postselection that can be equivalently decided by any all-but-$\left \lfloor{n/2}\right \rfloor$ parties is valid for verification of genuine $n$-partite nonlocality. 
\end{theorem}
 
{\it Applications.---}
As mentioned above, our findings for $n=3$ can be applied to setups where the number of particles is conserved. This is similar to the findings of Refs.~\cite{zukowski2000,blasiak2020} because undesirable events come in pairs. We now apply our results to the YS proposal of Ref.~\cite{yurke1992a}. 
The corresponding setup for two parties \cite{yurke1992b} makes use of valid local postselection \cite{sciarrino2011,cabello2009}. For three parties, a local postselection is not sufficient to violate Bell inequalities.

The setup of Ref.~\cite{yurke1992a} is shown in Fig.~\ref{fig:yurke}. Three independent sources ($S1$, $S2$ and $S3$) emit a single photon. Each photon enters a beam splitter whose outcoming modes are directed to two measurement parties. At each party, the two incoming modes enter a second beam splitter after which they are measured with photon-counting detectors. Additionally, each party chooses a measurement setting by imprinting a phase ($\phi_A$, $\phi_B$ and $\phi_C$) in one of the incoming modes. Each party $P$ can either detect no photon ($0_P$), a single photon in the left ($l_P$) or right ($r_P$) detector, or two photons in the left ($l_P^2$) or right ($r_P^2$) detector. For perfectly indistinguishable photons, events with a detection both detectors destructively interfere \cite{hong1987}.  

\begin{figure}[b]
    \centering
    \includegraphics[width=.9\linewidth]{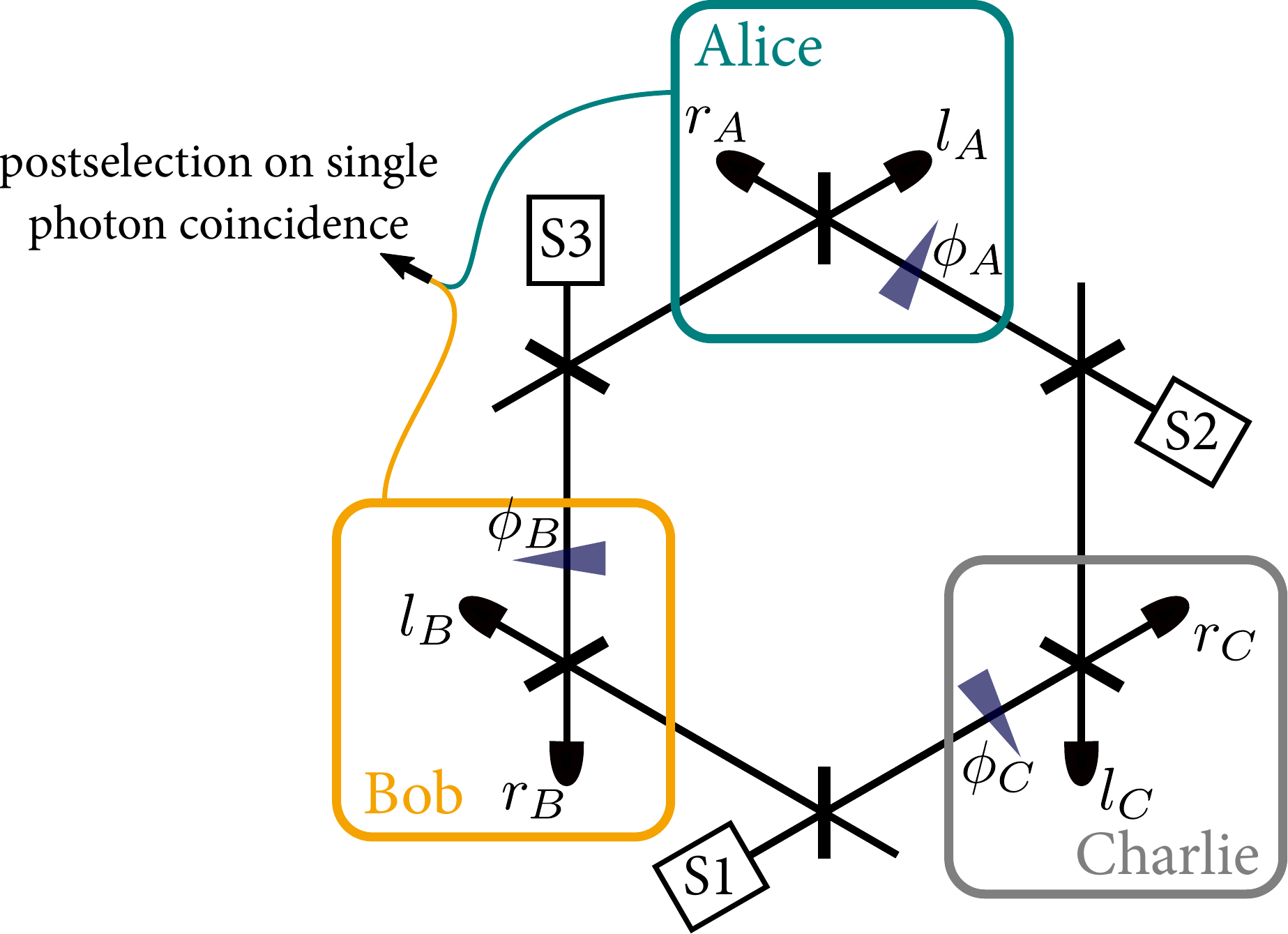}
    \caption{Proposal by Yurke and Stoler \cite{yurke1992a}: Three independent photon sources emit photons that are distributed via beam splitters among Alice, Bob and Charlie. Each party imprints a local phase in one of the incoming arms and measures the two modes after applying a second beam splitter. After postselecting events that show a single detection per party, the statistics show genuinely nonlocal (GHZ-like) features. Since the postselection can be equivalently decided after exclusion of any party, it represents a valid test for Bell inequality violations.}
    \label{fig:yurke}
\end{figure}

Assuming unit detection efficiency, the observed events can be divided into two groups: either each party receives a photon, or one party does not detect a photon and another detects two ($D$). The first group can be further divided into an even ($E$) or odd ($O$) number of right detector clicks. Depending on the total phase $\phi=(\phi_A+\phi_B+\phi_C)/2$, the probability $P(e)$ of an event $e$ is $P(e\in E)=\cos^2(\phi)/16$, $P(e\in O)=\sin^2(\phi)/16$ and $P(e\in D)=1/32$ \cite{yurke1992a}. To observe GMN, we must postselect events in $E$ and $O$. This postselection can be decided by any two parties and, according to Theorem 1, is valid to verify GMN. Indeed, say Alice measures two observables $x_i$ ($i=1,2$) resulting in outcomes $a=1$ ($a=-1$) for Alice's observation of $r_A$ ($l_A$), similarly for Bob and Charlie. Hybrid hidden-variable models \eqref{eq:hybrid} fulfill the three-partite-nonlocality-testing Svetlichny inequality \cite{svetlichny1987} 
\begin{align}
    I = \big| \langle x_1y_1z_1 \rangle + \langle x_1y_1z_2 \rangle + \langle x_1y_2z_1 \rangle + \langle x_2y_1z_1 \rangle \notag \\ 
    -  \langle x_1y_2z_2 \rangle - \langle x_2y_1z_2 \rangle - \langle x_2y_2z_1 \rangle - \langle x_2y_2z_2 \rangle\big|\leq 4 \label{eq:svetlichny},
\end{align}
where $\langle \dots \rangle$ denotes a statistical average. Let Alice choose between $\phi_A=0$ ($x_1$) and $\phi_A=-\pi/2$ ($x_2$), Bob between $\phi_B=\pi/4$ ($y_1$) and $\phi_B=-\pi/4$ ($y_2$) and Charlie between $\phi_C=0$ ($z_1$) and $\phi_C=-\pi/2$ ($z_2$). Then the Svetlichny inequality \eqref{eq:svetlichny} is maximally violated by the postselected statistics, $I=4\sqrt{2}$ \cite{mitchell2004}. Note that, in the hybrid hidden-variable model, we allow for a source of classical shared randomness between all three parties.

In the case $n>3$, a conservation of the number of particles is not sufficient that postselection can be decided by all-but-$\left \lfloor{n/2}\right \rfloor$ parties. Further constraints or conservation laws imposed on the possible events are required.

In experiments, finite detection efficiencies open an additional loophole, the detection loophole \cite{pearle1970,brunner2014,scarani2019book}. In schemes such as the YS proposal that fulfill the all-but-one principle for perfect efficiencies, for realistic efficiencies, the postselection cannot be decided any two parties anymore. Commonly, the detection loophole is circumvented by the fair-sampling assumption \cite{brunner2014,scarani2019book} that the detection of incoming particles does not depend on the measurement setting of the detector. The detection loophole can be rigorously closed by sharpening the Bell inequalities \cite{larsson1998} or taking into account all observed events \cite{sciarrino2011}. An application of these methods to the YS proposal is beyond the scope of this work.

{\it Conclusions.---}
We have introduced postselection strategies beyond local postselection such that the postselected statistics can validly be used to examine Bell inequality violation and verification of GMN. 
In the $n$-partite Bell scenario, we have shown that postselected statistics represent valid tests of multipartite Bell inequalities if the postselection can be equivalently decided by any all-but-$\left \lfloor{n/2}\right \rfloor$ parties. 
Thus, certain partially collaborative postselection strategies do not hinder the certification of GMN. 
Furthermore, the probability of successful postselection generally can be arbitrarily small.
For three parties, our results reduce to an all-but-one principle and can be applied to setups where the total number of particles is conserved. 
Particularly, for the proposal by Yurke and Stoler \cite{yurke1992a}, the postselected statistics are shown to maximally violate the Svetlichny inequality, demonstrating the creation of GMN from independent particle sources. 
Our results crucially facilitate the development of future quantum technologies due to the key role played by GMN.
The explicit use of causal diagrams in the proofs highlights their potential as a new tool in quantum and general physics.

\begin{acknowledgments}
{\it Acknowledgments.---}
This work was supported by the European Commission through the QuantERA ERA-NET Cofund in Quantum Technologies project “CEBBEC”.
\end{acknowledgments}

\clearpage
\onecolumngrid

\begin{center}
{\large \bf Genuine Multipartite Nonlocality with Causal-Diagram Postselection: \\
Supplemental Material  }
\end{center}

\begin{center}
\begin{minipage}{.85\linewidth}
 In this Supplementary Material, we first give further details and causal diagrams to support the proof of condition \textbf{I} of Theorem 1 in the main text. Then, we provide the proof and some discussion of Theorem 2 in the main text.
\end{minipage}
\end{center}

\section{Additional details for the proof of Theorem 1}

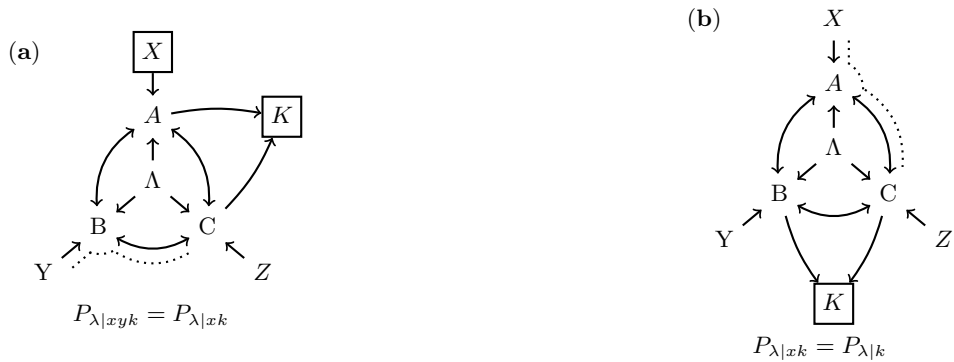
\begin{figure}[b]
\begin{minipage}{.49\linewidth}
\centering 

\begin{tikzpicture}[thick]
\centering

\node[draw,rectangle,text height=7pt,text depth=2pt] (X) {$X$};
\node[text height=7pt,text depth=2pt,left=1.1cm of X] (label) {$\mathbf{(a)}$};
\node[text height=7pt,text depth=2pt,below=.3cm of X] (A) {$A$};
\node[text height=7pt,text depth=2pt,below=.3cm of A] (Lambda) {$\Lambda$};
\node[text height=7pt,text depth=2pt,below left=.05cm and 0.25cm of Lambda] (B) {B};
\node[text height=7pt,text depth=2pt,below left=.05cm and 0.25cm of B] (Y) {Y};
\node[text height=7pt,text depth=2pt,below right=.05cm and 0.25cm of Lambda] (C) {C};
\node[text height=7pt,text depth=2pt,below right=.05cm and 0.25cm of C] (Z) {$Z$};
\node[draw,rectangle,text height=7pt,text depth=2pt,right=1.2cm of A] (K) {$K$};

\draw[->] (Lambda) edge (A)
          (Lambda) edge (B)
          (Lambda) edge (C)
          (X) edge (A)
          (Y) edge (B)
          (Z) edge (C)
          (A) edge[bend left=10] (K)
          (C) edge[bend right=10] (K)
;
\draw[<->]
          (A) edge[bend right=30]  (B)
          (A) edge[bend left=30]  (C)
          (B) edge[bend right=30]  (C)
;


\draw[-]
          (B) edge[dotted,bend right=30,transform canvas={yshift=-2mm}]  (C)
          (Y) edge[dotted,transform canvas={xshift=1.4mm,yshift=-1.3mm}] (B)
;
\draw[dotted] (B)+(-.12,-.33) arc[start angle=250, end angle=300, radius=4mm]
;
\end{tikzpicture}

    $P_{\lambda|xyk}=P_{\lambda|xk}$
\end{minipage}
\begin{minipage}{.49\linewidth}
\centering 

\begin{tikzpicture}[thick]
\centering

\node[text height=7pt,text depth=2pt] (X) {$X$};
\node[text height=7pt,text depth=2pt,left=1.1cm of X] (label) {$\mathbf{(b)}$};
\node[text height=7pt,text depth=2pt,below=.3cm of X] (A) {$A$};
\node[text height=7pt,text depth=2pt,below=.3cm of A] (Lambda) {$\Lambda$};
\node[text height=7pt,text depth=2pt,below left=.05cm and 0.25cm of Lambda] (B) {B};
\node[text height=7pt,text depth=2pt,below left=.05cm and 0.25cm of B] (Y) {Y};
\node[text height=7pt,text depth=2pt,below right=.05cm and 0.25cm of Lambda] (C) {C};
\node[text height=7pt,text depth=2pt,below right=.05cm and 0.25cm of C] (Z) {$Z$};
\node[draw,rectangle,text height=7pt,text depth=2pt,below=1.5cm of Lambda] (K) {$K$};

\draw[->] (Lambda) edge (A)
          (Lambda) edge (B)
          (Lambda) edge (C)
          (X) edge (A)
          (Y) edge (B)
          (Z) edge (C)
          (C) edge[bend left=10] (K)
          (B) edge[bend right=10] (K)
;
\draw[<->]
          (A) edge[bend right=30]  (B)
          (A) edge[bend left=30]  (C)
          (B) edge[bend right=30]  (C)
;

\draw[-]
          (A) edge[dotted,bend left=30,transform canvas={xshift=1.7mm,yshift=.9mm}]  (C)
          (X) edge[dotted,transform canvas={xshift=2mm}] (A)
;
\draw[dotted] (A)+(.19,.28) arc[start angle=60, end angle=-5, radius=4mm]
;

\end{tikzpicture}

    $P_{\lambda|xk}=P_{\lambda|k}$
\end{minipage}

   \caption{Causal diagrams that are used in the proof of condition \textbf{I} of Theorem 1 in the main text, to show conditional dependence of (a) $Y$ and $\Lambda$ and of (b) $X$ and $\Lambda$. Below the diagrams, we note the conditional independence relation that is demonstrated. We indicate some of the no-signalling conditions with a dotted line.}
    \label{fig3:proof3}
\end{figure}

In the proof of Theorem 1 in the main text, we have to show the free will condition
\begin{equation}
    \mathbf{I} \quad P_{\lambda|xyzk}=P_{\lambda|k}. \notag
\end{equation}
In the main text, we only show the first step of this proof, $P_{\lambda|xyzk}=P_{\lambda|xyk}$, by means of the causal diagram depicted in Fig.~2(c) in the main text. Here, we provide the causal diagrams that are used to prove the two other steps needed for condition \textbf{I}. 

To prove that $P_{\lambda|xyk}=P_{\lambda|xk}$, we have to use a causal diagram where the postselection $K$ is decided by Alice and Charlie, c.f. Fig.~\ref{fig3:proof3}(a). This diagram is equally valid as the one of Fig.~2(c) in the main text because, by assumption, the postselection $K$ can be decided by \emph{any} two parties. In contrast to Fig.~2(c) in the main text, we also only condition on $K$ and on only one further measurement choice ($X$). With the help of the diagram in Fig.~\ref{fig3:proof3}(a), we see that all paths between $Z$ and $\Lambda$ are blocked according to the d-separation rules: the path $Y\rightarrow B \rightarrow \Lambda$ is blocked because $B$ is a collider that is not conditioned on. Furthermore, the path $Y\rightarrow B \rightarrow C \rightarrow K \rightarrow A \rightarrow \Lambda$ is blocked because of the no-signalling condition that forbids any causal influence from $Y$ to $C$ [indicated as dotted line in Fig.~\ref{fig3:proof3}(a)]. The same reasoning holds true for the path $Y\rightarrow B \rightarrow A \rightarrow K \rightarrow C \rightarrow \Lambda$. 

Finally, to prove $P_{\lambda|xk}=P_{\lambda|k}$, we use a causal diagram where the postselection $K$ is decided by Bob and Charlie, c.f. Fig.~\ref{fig3:proof3}(b). In contrast to Fig.~2(c) in the main text and to Fig.~\ref{fig3:proof3}(a), we only condition on $K$. Again, the path $X\rightarrow A \rightarrow \Lambda$ is blocked because $A$ is a collider that is not conditioned on. Furthermore, similar as above, the paths $X\rightarrow A \rightarrow C \rightarrow K \rightarrow B \rightarrow \Lambda$ and $X\rightarrow A \rightarrow B \rightarrow K \rightarrow C \rightarrow \Lambda$ are blocked due to the no-signalling conditions.

\section{Proof of theorem 2}

Here, we prove Theorem 2 of the main text. In particular, we consider the case of an $n$-partite Bell scenario and show that a postselection that can be decided by (any) all but $\left \lfloor{n/2}\right \rfloor$ parties ($\left \lfloor{\cdot}\right \rfloor$ is the floor function) is valid to demonstrate genuine $n$-partite nonlocality. We assume that the $l$th party can choose a binary measurement setting $x_l$ and obtains an outcome $a_l$. The hybrid local-nonlocal hidden-variable model is now a mixture of all different possibilities of describing the joint probability distribution with nonlocal resources excluding genuine $n$-partite nonlocality. Formally, the hybrid model dictates that the joint probability distribution is of the form 
\be
P_{a_1\dots a_n|x_1 \dots x_n} =\sum_j\sum_{\lambda_j\in\Lambda_j}P_{\lambda_j}P_{a_1\dots a_n|x_1 \dots x_n\lambda_j},
\ee
where for each $j$, $P_{a_1\dots a_n|x_1 \dots x_n\lambda_j}$ has to factorize into some bipartition, for instance,  
\be 
P_{a_1\dots a_n|x_1 \dots x_n\lambda_j} = P_{a_1\dots a_l|x_1 \dots x_l\lambda_j}P_{a_{l+1}\dots a_n|x_{l+1} \dots x_n\lambda_j}. \label{eq:factor_n}
\ee
Furthermore, the no-signalling conditions $P_{a_l|x_1 \dots x_n}=P_{a_l|x_l}$ have to be fulfilled. 

To examine whether a given postselection strategy is valid to demonstrate genuine $n$-partite nonlocality, we have to check if the postselection preserves the factorization structure of the hybrid local-nonlocal model. For this purpose, we first use the chain rule to write 
\be
P_{a_1\dots a_n|x_1 \dots x_n k} =\sum_j\sum_{\lambda_j\in\Lambda_j}P_{\lambda_j|x_1 \dots x_n k}P_{a_1\dots a_n|x_1 \dots x_n \lambda_j k}.
\ee
We see that to secure a valid postselection strategy, we have to prove conditions similar the ones that we had to prove for Theorem 1 of the main text: the postselected statistics must fulfill the free-will condition (cf. condition $\mathbf{I}$ in main text) 
\begin{equation}
    P_{\lambda |x_1 x_2 \dots x_n k } = P_{\lambda|k}\label{eq:freewill}
\end{equation}
and the different subensembles must still factorize as in Eq. \eqref{eq:factor_n} after postselection (cf. conditions $\mathbf{II}$ in main text), 
\be 
P_{a_1\dots a_n|x_1 \dots x_n\lambda_j k} = P_{a_1\dots a_l|x_1 \dots x_l\lambda_jk}P_{a_{l+1}\dots a_n|x_{l+1} \dots x_n\lambda_jk}, \label{eq:condition_factor_n}
\ee
with similar conditions for the other subensembles.

Now assume that the postselection $K$ can be decided by any $n-\left \lfloor{n/2}\right \rfloor$ (all but $\left \lfloor{n/2}\right \rfloor$) parties. This implies that any causal diagram where $K$ is only influenced by $n-\left \lfloor{n/2}\right \rfloor$ parties (i.e., connected by an arrow from these parties' measurement results), is a valid causal diagram that can be used in the proof. For instance, in Fig.~\ref{fig:n_partite}(a), $K$ is decided by the $n-\left \lfloor{n/2}\right \rfloor$ last parties. 

\begin{figure}[b]
\begin{minipage}{.45\linewidth}
\begin{tikzpicture}[thick]
\centering 

\node[text height=7pt,text depth=2pt] (A1) {$A_1$};
\node[text height=7pt,text depth=2pt,right=.3cm of A1] (dotsleft) {$\cdots$};
\node[text height=7pt,text depth=2pt,right=.3cm of dotsleft] (Ak) {$A_{\left \lfloor{n/2}\right \rfloor}$};
\node[right=.1cm of Ak] (center) {};
\node[text height=7pt,text depth=2pt,right=.1cm of center] (Akplus1) {$A_{\left \lfloor{n/2}\right \rfloor+1}$};
\node[text height=7pt,text depth=2pt,right=.3cm of Akplus1] (dotsright) {$\cdots$};
\node[text height=7pt,text depth=2pt,right=.3cm of dotsright] (An) {$A_n$};

\node[text height=7pt,text depth=2pt,below=1.1cm of center] (lambda) {$\Lambda$};
\node[draw,rectangle,text height=7pt,text depth=2pt,above=.6cm of center] (K) {$K$};
\node[text height=7pt,text depth=2pt,left=2.5cm of K] (label) {$\mathbf{(a)}$};

\node[text height=7pt,text depth=2pt,below left = .2cm and -0.1cm of A1] (X1) {$X_1$};
\node[draw,rectangle,text height=7pt,text depth=2pt,below left=.2cm and -0.3cm of Ak] (Xk) {$X_{\left \lfloor{n/2}\right \rfloor}$};
\node[draw,rectangle,text height=7pt,text depth=2pt,below right=.2cm and -0.6cm of Akplus1] (Xkplus1) {$X_{\left \lfloor{n/2}\right \rfloor+1}$};
\node[draw,rectangle,text height=7pt,text depth=2pt,below right=.2cm and -0.1cm of An] (Xn) {$X_n$};

\draw[densely dashed, rounded corners, thin]
($(A1)+(-0.3,0.2)$) rectangle ($(An)+(0.3,-0.25)$);

\draw[->] (X1) edge (A1)
          (Xk) edge (Ak)
          (Xkplus1) edge (Akplus1)
          (Xn) edge (An)
          (lambda) edge[bend right=45] (An)
          (lambda) edge (Akplus1)
          (lambda) edge(Ak)
          (lambda) edge[bend left=45] (A1)
          (An) edge[bend right=15] (K)
          (Akplus1) edge (K)

;

\end{tikzpicture}\\
$P_{\lambda |x_1 x_2 \dots x_n k } = P_{\lambda |x_2 x_3\dots x_n k}$

\end{minipage}
\hspace{.7cm}
\begin{minipage}{.45\linewidth}
\begin{tikzpicture}[thick]
\centering 

\node[text height=7pt,text depth=2pt] (A1) {$A_1$};
\node[text height=7pt,text depth=2pt,right=.3cm of A1] (dotsleft) {$\cdots$};
\node[text height=7pt,text depth=2pt,right=.3cm of dotsleft] (Ak) {$A_{\left \lfloor{n/2}\right \rfloor}$};
\node[right=.1cm of Ak] (center) {};
\node[text height=7pt,text depth=2pt,right=.1cm of center] (Akplus1) {$A_{\left \lfloor{n/2}\right \rfloor+1}$};
\node[text height=7pt,text depth=2pt,right=.3cm of Akplus1] (dotsright) {$\cdots$};
\node[text height=7pt,text depth=2pt,right=.3cm of dotsright] (An) {$A_n$};

\node[draw,rectangle,text height=7pt,text depth=2pt,below=1.1cm of center] (lambda) {$\Lambda_{j_0}$};
\node[draw,rectangle,text height=7pt,text depth=2pt,above=.6cm of center] (K) {$K$};
\node[text height=7pt,text depth=2pt,left=2.5cm of K] (label) {$\mathbf{(b)}$};

\node[draw,rectangle,text height=7pt,text depth=2pt,below left = .2cm and -0.1cm of A1] (X1) {$X_1$};
\node[draw,rectangle,text height=7pt,text depth=2pt,below left=.2cm and -0.3cm of Ak] (Xk) {$X_{\left \lfloor{n/2}\right \rfloor}$};
\node[draw,rectangle,text height=7pt,text depth=2pt,below right=.2cm and -0.6cm of Akplus1] (Xkplus1) {$X_{\left \lfloor{n/2}\right \rfloor+1}$};
\node[draw,rectangle,text height=7pt,text depth=2pt,below right=.2cm and -0.1cm of An] (Xn) {$X_n$};

\draw[densely dashed, rounded corners, thin]
($(A1)+(-0.3,0.2)$) rectangle ($(Ak)+(0.5,-0.25)$);
\draw[densely dashed, rounded corners, thin]
($(Akplus1)+(-0.7,0.2)$) rectangle ($(An)+(0.3,-0.25)$);

\draw[->] (X1) edge (A1)
          (Xk) edge (Ak)
          (Xkplus1) edge (Akplus1)
          (Xn) edge (An)
          (lambda) edge[bend right=45] (An)
          (lambda) edge (Akplus1)
          (lambda) edge(Ak)
          (lambda) edge[bend left=45] (A1)
          (An) edge[bend right=15] (K)
          (Akplus1) edge (K)
;

\end{tikzpicture}\\
$P_{a_1\dots a_{\left \lfloor{n/2}\right \rfloor}|a_{\left \lfloor{n/2}\right \rfloor+1}\dots a_n x_1 \dots x_n\lambda_{j_0}}=P_{a_1\dots a_{\left \lfloor{n/2}\right \rfloor}|x_1 \dots x_n\lambda_{j_0} k}$

\end{minipage}
    \caption{Different causal diagrams in the $n$-partite Bell scenario. (a) The causal diagram that is used in the first step of the proof that the postselected statistics fulfills the free will condition, Eq.~\eqref{eq:freewill}. (b) The causal diagram that is used to show Eq.~\eqref{eq:cond2second}, contributing to the proof of the factorization condition, Eq.~\eqref{eq:condition_factor_n}, for $l=\left \lfloor{n/2}\right \rfloor$. The dashed boxes indicate possible nonlocal correlations that are shared between (different subgroups of) the parties. The no-signalling conditions are implied.}
    \label{fig:n_partite}
\end{figure}
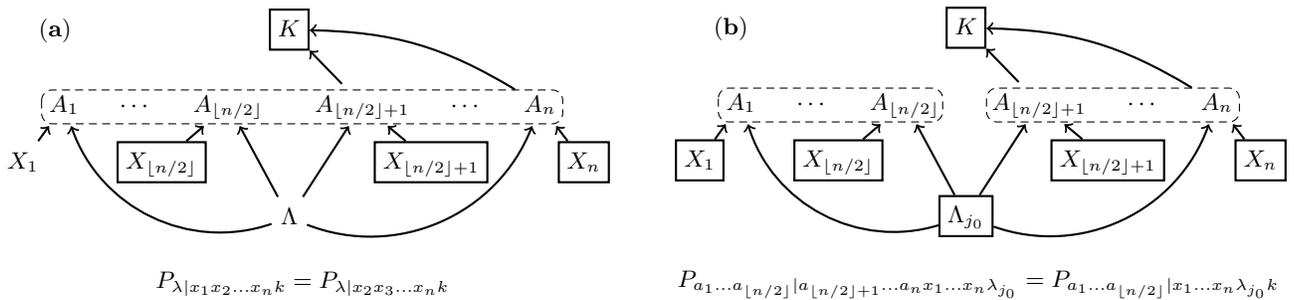

The causal diagram shown in Fig.~\ref{fig:n_partite}(a) describes the first step of the proof of the free will condition. Note that here, in contrast to the causal diagrams in Fig.~2 of the main text, we symbolically indicate possible nonlocal correlations between outcomes with dashed boxes. No-signalling conditions are implied. The shown diagram shows that we have 
$P_{\lambda |x_1 x_2 \dots x_n k } = P_{\lambda |x_2 x_3\dots x_n k}$ for the following two observations: (i) the path $X_1\rightarrow A_1 \leftarrow \Lambda$ is blocked because $A_1$ is a collider, and (ii) all other paths have to be of the form $X_1\rightarrow A_1 \rightarrow A_l\rightarrow \dots $ for $l\neq 1$. The latter paths are blocked due to the no-signalling condition that forbids causal influences from $X_1$ to $A_l$. As described in the proof of Theorem 2, note that the above independence also holds if one removes the conditioning on (parts of) the other measurement settings. Then, conditional independence of $\Lambda$ from the other $X_l$ (with a corresponding conditioning on $K$ and parts of the remaining measurement settings) can be shown in a similar way but by making use of different causal diagrams: to show that $\Lambda$ and $X_l$ are conditionally independent, one has to consider a causal diagram where $K$ is not decided by $A_l$. Finally, we find that the free-will condition, Eq.~\eqref{eq:freewill},
is fulfilled. Analogously, one can check that the no-signalling principle also holds for the postselected statistics. We want to note that until now, the proof remains valid even if the postselection can only be decided by all but one parties.

Next, we consider the factorization conditions, Eq.~\eqref{eq:condition_factor_n}. We first take the subensemble $\Lambda_{j_0}$ that corresponds to a factorization as in Eq.~\eqref{eq:factor_n} with $l=\left \lfloor{n/2}\right \rfloor$. It generally holds that 
\be 
P_{a_1\dots a_n|x_1 \dots x_n\lambda_{j_0} k} = P_{a_1\dots a_{\left \lfloor{n/2}\right \rfloor}|a_{\left \lfloor{n/2}\right \rfloor+1}\dots a_n x_1 \dots x_n\lambda_{j_0} k}P_{a_{\left \lfloor{n/2}\right \rfloor+1}\dots a_n| x_1 \dots x_n\lambda_{j_0} k}.\label{eq:cond2first}
\ee
Then, we use the causal diagram shown in Fig.~\ref{fig:n_partite}(b), where we restrict ourselves on the subensemble $\Lambda_{j_0}$ and where $K$ is decided by the $n-\left \lfloor{n/2}\right \rfloor$ parties $A_{\left \lfloor{n/2}\right \rfloor+1},\dots, A_n$, to write 
\be 
P_{a_1\dots a_{\left \lfloor{n/2}\right \rfloor}|a_{\left \lfloor{n/2}\right \rfloor+1}\dots a_n x_1 \dots x_n\lambda_{j_0}}=P_{a_1\dots a_{\left \lfloor{n/2}\right \rfloor}|x_1 \dots x_n\lambda_{j_0} k}.\label{eq:cond2second}
\ee 
This equation holds because the only possible paths that connect $A_1,\dots, A_{\left \lfloor{n/2}\right \rfloor}$ to $A_{\left \lfloor{n/2}\right \rfloor+1},\dots, A_n$ pass through the node $\Lambda_{j_0}$ that is conditioned on, such that the paths are blocked. Note that, in Fig.~\ref{fig:n_partite}(b), the nonlocal influences (indicated as dashed boxes) are confined to the two groups $A_1,\dots, A_{\left \lfloor{n/2}\right \rfloor}$ and $A_{\left \lfloor{n/2}\right \rfloor+1},\dots, A_n$, respectively, as dictated by Eq.~\eqref{eq:factor_n}.

By using the same diagram as in Fig.~\ref{fig:n_partite}(b) but conditioning on $\Lambda_{j_0}$, $K$, and $X_1,\dots, X_{\left \lfloor{n/2}\right \rfloor}$, one shows that 
\be 
P_{a_1\dots a_{\left \lfloor{n/2}\right \rfloor}|x_1 \dots x_n\lambda_{j_0} k}=P_{a_1\dots a_{\left \lfloor{n/2}\right \rfloor}|x_1 \dots x_{\left \lfloor{n/2}\right \rfloor}\lambda_{j_0} k},
\ee
because all possible paths again pass through the conditioned node $\Lambda_{j_0}$ and are thus blocked. Finally, the same diagram with conditioning on $\Lambda_{j_0}$, $K$, and $X_{\left \lfloor{n/2}\right \rfloor+1},\dots, X_n$ proves 
\be 
P_{a_{\left \lfloor{n/2}\right \rfloor+1}\dots a_n| x_{1} \dots x_n\lambda_{j_0} k} = P_{a_{\left \lfloor{n/2}\right \rfloor+1}\dots a_n| x_{\left \lfloor{n/2}\right \rfloor+1} \dots x_n\lambda_{j_0} k}\label{eq:cond2last}. 
\ee
By combining Eqs.~(\ref{eq:cond2first}-\ref{eq:cond2last}), we obtain the factorization condition, Eq.~\eqref{eq:condition_factor_n}, for $l=\left \lfloor{n/2}\right \rfloor$. 

At this point, one can directly see why the factorization conditions only hold if the postselection can be decided by all but $\left \lfloor{n/2}\right \rfloor$ parties. If any additional party has to be included to decide the postselection $K$, there has to be at least one connection from the group $A_1,\dots, A_{\left \lfloor{n/2}\right \rfloor}$ to $K$, see, e.g., Fig.~\ref{fig:n_partite2}(a). Thus, Eq.~\eqref{eq:cond2second} generally does not hold because the two groups $A_1,\dots, A_{\left \lfloor{n/2}\right \rfloor}$ and $A_{\left \lfloor{n/2}\right \rfloor+1},\dots, A_n$ are connected by a path through the collider $K$. This path is unblocked according to the d-separation rules (see main text) because the collider $K$ is conditioned on. In the main text, we have exemplified this problem in the case of $n=4$, where the diagram in Fig.~\ref{fig:n_partite2}(a) consists of two pairs of nonlocality-sharing parties, and a postselection that has to be decided by three parties necessarily connects the two pairs. Thus, we see why a general postselection that is decided with an all-but-one principle is not valid to demonstrate genuine $n$-partite nonlocality for $n>3$. 

Finally, we note that other subensembles $\Lambda_j$ that factorize as in Eq.~\eqref{eq:factor_n} with $l<\left \lfloor{n/2}\right \rfloor$, also factorize according to Eq.~\eqref{eq:condition_factor_n}. The only difference to the above proof is that, in the causal diagram of Fig.~\ref{fig:n_partite}(b), the postselection $K$ can be decided by merely a subgroup of the parties $A_{l+1},\dots, A_n$, see Fig.~\ref{fig:n_partite2}(b). Thus, all above steps of the proof remain valid. All remaining subensembles $\Lambda_j$ fulfill their corresponding factorization condition as well. This can be seen by considering the appropriate causal diagrams, where the postselection $K$ must be decided by (parts of) the larger group of the specific bipartition of interest, c.f. Eq.~\eqref{eq:factor_n}. This largest group always contains $r\geq n-\left \lfloor{n/2}\right \rfloor$ parties and is therefore sufficient to decide the postselection.

\begin{figure}[t]
\begin{minipage}{.45\linewidth}
\begin{tikzpicture}[thick]
\centering 

\node[text height=7pt,text depth=2pt] (A1) {$A_1$};
\node[text height=7pt,text depth=2pt,right=.3cm of A1] (dotsleft) {$\cdots$};
\node[text height=7pt,text depth=2pt,right=.3cm of dotsleft] (Ak) {$A_{\left \lfloor{n/2}\right \rfloor}$};
\node[right=.1cm of Ak] (center) {};
\node[text height=7pt,text depth=2pt,right=.1cm of center] (Akplus1) {$A_{\left \lfloor{n/2}\right \rfloor+1}$};
\node[text height=7pt,text depth=2pt,right=.3cm of Akplus1] (dotsright) {$\cdots$};
\node[text height=7pt,text depth=2pt,right=.3cm of dotsright] (An) {$A_n$};

\node[draw,rectangle,text height=7pt,text depth=2pt,below=1.1cm of center] (lambda) {$\Lambda_{j_0}$};
\node[draw,rectangle,text height=7pt,text depth=2pt,above=.6cm of center] (K) {$K$};
\node[text height=7pt,text depth=2pt,left=2.5cm of K] (label) {$\mathbf{(a)}$};

\node[draw,rectangle,text height=7pt,text depth=2pt,below left = .2cm and -0.1cm of A1] (X1) {$X_1$};
\node[draw,rectangle,text height=7pt,text depth=2pt,below left=.2cm and -0.3cm of Ak] (Xk) {$X_{\left \lfloor{n/2}\right \rfloor}$};
\node[draw,rectangle,text height=7pt,text depth=2pt,below right=.2cm and -0.6cm of Akplus1] (Xkplus1) {$X_{\left \lfloor{n/2}\right \rfloor+1}$};
\node[draw,rectangle,text height=7pt,text depth=2pt,below right=.2cm and -0.1cm of An] (Xn) {$X_n$};

\draw[densely dashed, rounded corners, thin]
($(A1)+(-0.3,0.2)$) rectangle ($(Ak)+(0.5,-0.25)$);
\draw[densely dashed, rounded corners, thin]
($(Akplus1)+(-0.7,0.2)$) rectangle ($(An)+(0.3,-0.25)$);

\draw[->] (X1) edge (A1)
          (Xk) edge (Ak)
          (Xkplus1) edge (Akplus1)
          (Xn) edge (An)
          (lambda) edge[bend right=45] (An)
          (lambda) edge (Akplus1)
          (lambda) edge(Ak)
          (lambda) edge[bend left=45] (A1)
          (An) edge[bend right=15] (K)
          (Akplus1) edge (K)
          (Ak) edge (K)
;

\end{tikzpicture}

\end{minipage}
\hspace{.7cm}
\begin{minipage}{.45\linewidth}
\begin{tikzpicture}[thick]
\centering 

\node[text height=7pt,text depth=2pt] (A1) {$A_1$};
\node[text height=7pt,text depth=2pt,right=.3cm of A1] (dotsleft) {$\cdots$};
\node[text height=7pt,text depth=2pt,right=.3cm of dotsleft] (Ak) {$A_{l}$};
\node[right=.1cm of Ak] (center) {};
\node[text height=7pt,text depth=2pt,right=.1cm of center] (Akplus1) {$A_{l+1}$};
\node[text height=7pt,text depth=2pt,right=.3cm of Akplus1] (dotsright) {$\cdots$};
\node[text height=7pt,text depth=2pt,right=.3cm of dotsright] (An) {$A_n$};

\node[draw,rectangle,text height=7pt,text depth=2pt,below=1.1cm of center] (lambda) {$\Lambda_{j_0}$};
\node[draw,rectangle,text height=7pt,text depth=2pt,above=.6cm of center] (K) {$K$};
\node[text height=7pt,text depth=2pt,left=2.5cm of K] (label) {$\mathbf{(b)}$};

\node[draw,rectangle,text height=7pt,text depth=2pt,below left = .2cm and -0.1cm of A1] (X1) {$X_1$};
\node[draw,rectangle,text height=7pt,text depth=2pt,below left=.2cm and -0.3cm of Ak] (Xk) {$X_{l}$};
\node[draw,rectangle,text height=7pt,text depth=2pt,below right=.2cm and -0.6cm of Akplus1] (Xkplus1) {$X_{l+1}$};
\node[draw,rectangle,text height=7pt,text depth=2pt,below right=.2cm and -0.1cm of An] (Xn) {$X_n$};

\draw[densely dashed, rounded corners, thin]
($(A1)+(-0.3,0.2)$) rectangle ($(Ak)+(0.3,-0.25)$);
\draw[densely dashed, rounded corners, thin]
($(Akplus1)+(-0.4,0.2)$) rectangle ($(An)+(0.3,-0.25)$);

\draw[->] (X1) edge (A1)
          (Xk) edge (Ak)
          (Xkplus1) edge (Akplus1)
          (Xn) edge (An)
          (lambda) edge[bend right=45] (An)
          (lambda) edge (Akplus1)
          (lambda) edge(Ak)
          (lambda) edge[bend left=45] (A1)
          (dotsright) edge[bend right=15] (K)
          (Akplus1) edge (K)
;

\end{tikzpicture}

\end{minipage}
    \caption{(a) A causal diagram showing that, if the postselection $K$ can only decided by $\left \lfloor{n/2}\right \rfloor+1$ parties, Eq.~\eqref{eq:cond2second} cannot be proven because there is an open path through the conditioned collider $K$. (b) The causal diagram that is used to show Eq.~\eqref{eq:cond2second}, contributing to the proof of the factorization condition, Eq.~\eqref{eq:condition_factor_n}, for $l<\left \lfloor{n/2}\right \rfloor$.}
    \label{fig:n_partite2}
\end{figure}
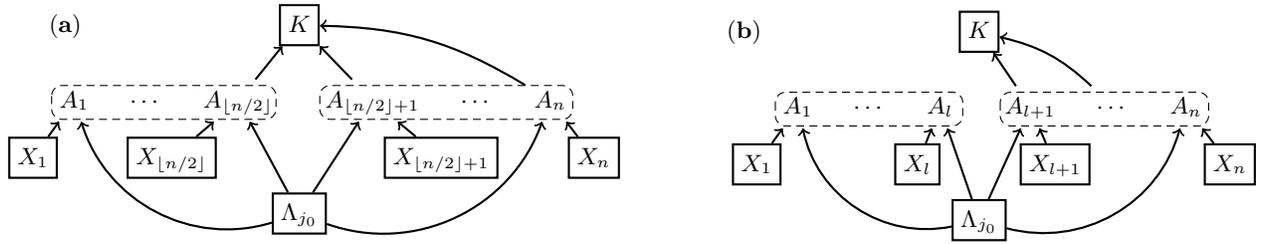

\end{document}